\documentclass[preprint,aps,prd,floatfix, showkeys,showpacs]{revtex4-2}
\usepackage{graphicx}
\usepackage{amsfonts}
\usepackage{amsmath}
\usepackage{rotating}
\usepackage{amssymb,amsthm}
\usepackage{xcolor}
\usepackage[normalem]{ulem}

\newtheorem{thm}{Theorem}[section]

\newtheorem{con}[thm]{Conjecture} 
\theoremstyle{definition} 
  
\theoremstyle{remark}  
  
\def\beq{\begin{eqnarray}}  
\def\eeq{\end{eqnarray}}  
\def\bsp{\begin{split}}  
\def\esp{\end{split}}

\begin{document}  

\title{Geometric surfaces: An invariant characterization of spherically symmetric black hole horizons and wormhole throats}

\author{D. D. McNutt} 
\email{david.d.mcnutt@uis.no}
\affiliation{Faculty of Science and Technology, University of Stavanger, N-4036 Stavanger, Norway}

\author{W. Julius}
\email{william$\_$julius1@baylor.edu} 
\author{M. Gorban}
\email{matthew$\_$gorban1@baylor.edu}
\author{B. Mattingly}
\email{brandon$\_$mattingly@baylor.edu}
\author{P. Brown}
\email{patrick.brown.ctr@afit.edu}
\author{G. Cleaver}
\email{gerald$\_$cleaver@baylor.edu}

\affiliation{Early Universe Cosmology and Strings (EUCOS) Group Center for Astrophysics, Space Physics  and Engineering Research (CASPER),  Baylor University,  Waco, TX 76798, USA}

\affiliation{ Department of Physics Baylor University,  Waco, TX 76798, USA}

%
%

\begin{abstract}   

We consider a spherically symmetric line element which admits either a black hole geometry or a wormhole geometry and show that in both cases the apparent horizon or the wormhole throat is partially characterized by the zero-set of a single curvature invariant. The detection of the apparent horizon by this invariant is consistent with the geometric horizon detection conjectures and implies that it is a geometric horizon of the black hole, while the detection of the wormhole throat presents a conceptual problem for the conjectures. To distinguish between these surfaces, we determine a set of curvature invariants that fully characterize the apparent horizon and wormhole throat. Motivated by this result, we introduce the concept of a geometric surface as a generalization of a geometric horizon and extend the geometric horizon detection conjectures to geometric surfaces. As an application, we employ curvature invariants to characterize three important surfaces of the line element introduced by Simpson, Martin-Moruno and Visser, which describes transitions between regular Vaidya black holes, traversable wormholes, and black bounces.    

\end{abstract} 



\maketitle

%

\section{Introduction}

Black holes and wormholes are fascinating and strange solutions of many gravitational theories. The defining feature of black hole solutions is the presence of an event horizon, the boundary of the region from which no information can be sent to a distant asymptotic external region \cite{Wald:1984rg}. In contrast, wormholes are objects that connect two or more distant regions of spacetime, or even separate universes, whose distinct regions are bounded by a wormhole throat \cite{morris88, Visser:1995cc}. While the interpretations of these two classes of solutions are drastically different, they do share a common property. They both admit a hypersurface that demarcates important regions of such solutions, i.e., the event horizon or the wormhole throat.

The standard definitions of both event horizons and wormhole throats can be difficult to implement in practice. The definition of the event horizon relies on the global structure of the spacetime and not just the local geometry.
Due to this, determining the location of an event horizon requires knowledge of the entire evolution of null curves. In particular, one must have knowledge not only of their local behaviour, but of the null curve's asymptotic future as well. Since the definition and determination of event horizons depends on the entire history and evolution of a spacetime, they are teleological objects \cite{Booth2005}.

Wormholes are not teleological; however, their definition also relies on global geometry. The topology of the entire spacetime manifold must be determined in order to distinguish between wormholes that connect distinct regions of the same spacetime and wormholes that connect two distinct spacetimes \cite{Hochberg:1998ha}. Since it is impossible for any observer to have global knowledge of a spacetime, it is natural to posit that there may be some local method of determining a quasi-local surface which will characterize and distinguish a black hole or wormhole.

For stationary black hole solutions, the event horizon can be identified with the Killing horizon, which is a local surface. However, for dynamical black hole solutions, the event horizon is not quasi-local and other quasi-local hypersurfaces must be used to define the boundary for such black holes. Two types of  quasi-local surfaces, marginally trapped tubes (MTTs) and trapping horizons (THs), have played an important role in the literature \cite{Booth2005, Senov}. In this paper, we will focus on an extension of MTTs known as future outer trapped horizons (FOTHs; also called apparent horizons) and these will be defined in the following subsection.  While apparent horizons and other quasi-local surfaces based on (marginally) trapped surfaces have been used in numerical general relativity \cite{thornburg2007event}, they are foliation dependent and therefore non-unique \cite{Booth2005}. 

The non-uniqueness of quasi-local horizons of black holes motivates the determination of invariant hypersurfaces that are independent of the choice of foliation. It has been conjectured that dynamical black holes admit a quasi-local hypersurface on which the curvature tensor and its covariant derivatives become more algebraically special. Such a hypersurface, called a {\it geometric horizon} (GH), can be invariantly defined by the vanishing of a particular set of curvature invariants \cite{ADA,Coley:2017vxb}. While several examples of dynamical black hole solutions have been shown to admit GHs, the exact definition of a GH has not been fully determined in terms of a particular set of curvature invariants. However, in the case of spherically symmetric black holes, the apparent horizon is a GH \cite{Coley:2017vxb} and hence is unique.

In a similar manner to the quasi-local horizons of black holes, the definition of a wormhole throat has an issue with foliation dependence. In the case of stationary or static wormhole solutions, the throat can be determined locally \cite{Hochberg:1997wp}. However, for dynamical wormholes there are several throat definitions that arise from imposing conditions on the expansion of ingoing or outgoing null geodesics and their derivatives \cite{Hochberg:1998ha, Hayward:2009yw, Tomikawa:2015swa}. As in the case of apparent horizons, such definitions of the wormhole throat depend on the foliation of spacetime and hence are non-unique. In this paper, we will consider the definition of a wormhole throat introduced by Hochberg and Viser \cite{Hochberg:1998ha} which describes the throat as the minimal two-surface where lightrays focus as they enter the surface and expand on the other side once they have passed through the throat \cite{Hochberg:1998ha}. We will call such surfaces Hochberg-Visser (HV) wormhole throats and define them in more detail in the following subsection. 


While apparent horizons play an important role in gravitational wave astronomy \cite{LIGO, Gralla:2017ufe}, wormholes remain theoretical objects since they require exotic matter that violates local energy conditions \cite{morris88, Visser:1995cc}. Apparent horizons and wormholes are equally important for the understanding of gravity and spacetime physics \cite{Hochberg:1997wp, hayward1994a} and have led to the refinement of definitions for special quasi-local surfaces in the literature \cite{Booth2005}. Since both apparent horizons and HV wormhole throats are defined in a similar manner using the expansion of null geodesics, it is natural to conjecture that there may be an invariant description of a wormhole throat similar to the GH conjecture. However, the current definition of a GH has two significant problems. 

The first problem is a matter of implementation. The definition of a GH is often given in terms of Cartan invariants \cite{kramer} which must be calculated in a certain prescribed coframe. In practice, determining such a coframe can be extremely difficult. It is possible that scalar polynomial curvature invariants (SPIs), which are truly frame independent, may be easier to use to define GHs in practice. SPIs have previously been used to examine the curvature structure of both black hole and wormhole solutions \cite{AbdelqaderLake2015, GANG, overduin2020curvature, Mattingly:2018vzl}.

The second problem is more significant and concerns the interpretation of these surfaces. For spherically symmetric solutions, the definition of a GH implies that a HV wormhole throat is a GH. From the GH conjectures, we would then erroneously conclude that a wormhole solution is in fact a black hole.  More generally, the definition of a GH cannot distinguish between Killing horizons, cosmological horizons (such as in the Reisner-Nordstrom (anti-) de Sitter solutions), or the ergosphere in Kerr metrics \cite{GANG}. Furthermore, a dynamical black hole solution may admit multiple GHs defined by the vanishing of distinct SPIs. For example, when a dynamical black hole solution is constructed by making a conformal transformation from a stationary black hole, the event horizon is determined by an SPI \cite{PageMcNutt,mcnutt2017}, but a second GH can also appear \cite{faraoni2011horizon}. Noting that the above surfaces are characterized as the zero set of curvature invariants, we will introduce the concept of a {\it geometric surface} (GS), which is a surface defined by the vanishing of a particular set of curvature invariants and further characterized by inequalities on additional non-zero curvature invariants. 

The outline of the paper is as follows: In the remainder of section I, we briefly define apparent horizons and HV wormhole throats. In section \ref{sec:1}, we introduce the line element for the imploding spherically symmetric spacetime and review the geometric preliminaries to define the apparent horizon, the GH and the HV  wormhole throat. In section \ref{sec:2}, we implement the Cartan-Karlhede (CK) algorithm to generate a set of Cartan invariants which define the GH and HV wormhole throat. In section \ref{sec:3}, we determine ratios of SPIs that characterize the GH and HV wormhole throat. In section \ref{sec:4}, we introduce an improved set of conjectures for GSs and illustrate the conjectures by considering the class of dynamical spherically symmetric metrics which describe the transitions between black bounces, regular Vaidya black holes, and traversable wormholes. In section \ref{sec:5}, we summarize our results and discuss future work. 

\subsection{Apparent horizons and wormhole throats}

In order to define these apparent horizons and wormhole throats, we introduce a complex-null frame basis $\{\ell, n, m, \bar{m}\}$ and its dual coframe $\{ n, \ell, \bar{m}, m \}$ such that the metric takes the form 
\beq g_{ab} = 2( m_{(a} \bar{m}_{b)} \label{gfrm} - \ell_{(a} n_{b)}) , \eeq
\noindent where round brackets around indices denote symmetrization. The expansions of the null directions $\theta_{(\ell)}$ and $\theta_{(n)}$ are then 
\beq \theta_{(\ell)} = \bar{q}^{\,ab} \nabla_a \ell_b \label{NullExpansionEll} \text{ and }
 \theta_{(n)} = \bar{q}^{\,ab} \nabla_a n_b \label{NullExpansionN}, \eeq
\noindent where $\bar{q}_{ab} = g_{ab} +2 \ell_{(a} n_{b)}$ is a local two-metric.

In the process of defining MTTs and THs, we will first introduce the concept of a {\it trapped surface}, $S$. This is a closed two-surface such that the expansions, $\theta_{(\ell)}$ and $ \theta_{(n)} $, along each of the two future-pointing null vectors normal to the surface, $\ell^a$ and $n^a$, are everywhere negative: 
\beq \theta_{(\ell)} < 0,\text{ and } \theta_{(n)} <0 \label{eqn:trappedS}. \eeq

\noindent In this case $\bar{q}_{ab}$ is the induced two metric on the surface $S$ for which $\ell^a$ and $n^a$ are normal vectors. As in equation \eqref{gfrm}, $\ell^a$ and $n^a$ are normalized to ensure they are outward/inward pointing null vector fields.

Using trapped surfaces, the original definition of an apparent horizon is given as the boundary of the union of all points that lie on some trapped surface \cite{P2}. Working with this definition is problematic in practice, as it relies on assumptions about asymptotic flatness and has problems with the smoothness of the apparent horizon \cite{Booth2005}. It is more common to work with {\it marginally trapped surfaces} (MTSs), which are closed two-surfaces for which the expansions, $\theta_{(\ell)}$ and $ \theta_{(n)} $, along each of the two future-pointing null vectors normal to the surface, satisfy: 
\beq \theta_{(\ell)} = 0,\text{ and } \theta_{(n)} <0 \label{eqn:MtrappedS}. \eeq

If the MTSs can be combined to produce a three-surface foliated by the MTSs, this is called a {\it marginally trapped tube} (MTT). If we additionally impose that $\mathcal{L}_n \theta_{(\ell)} <0$, then the MTT is a {\it future outer trapping horizon} (FOTH), which is frequently referred to as an  {\it apparent horizon} in the literature where it is used as a quasi-local analogue of a future event horizon \cite{hayward1994a, hayward1994b}. 
In summary, an apparent horizon is a smooth hypersurface that is foliated by MTSs such that the expansions relative to the ingoing and outgoing null foliation normal vector fields satisfy:

\beq \theta_{(\ell)} =0, ~\theta_{(n)} < 0, \text{ and } \mathcal{L}_n \theta_{(\ell)} = n^a \nabla_a \theta_{(\ell)} <0. \label{FOTH} \label{apparent def}\eeq

\noindent If the third condition is dropped, this becomes a dynamical horizon \cite{Booth2005}. 

The HV wormhole throat can be characterized in a similar manner using the expansion of ingoing and outgoing null geodesics. The HV wormhole throat is defined as the minimal two-surface where lightrays focus as they enter the mouth of the wormhole and expand on the other side once they have passed through the throat. This can be restated in terms of the expansion of the outward/inward pointing null vector fields which are normal to the surface: 
\beq \theta_{(\ell)} = 0 \text{~and~} \ell^a \nabla_a \theta_{(\ell)} \geq 0 \text{~or~} \theta_{(n)} = 0 \text{~and~} n^a \nabla_a \theta_{(n)} \geq 0  \label{HVthroat}.\eeq

\noindent In the cases considered here, it is not possible to have both expansions vanish simultaneously. Throughout this paper we will deliberately construct null frames such that $\ell^a$ is the null direction with vanishing expansion.


\section{Spherically symmetric wormholes and black holes} \label{sec:1}

The imploding spherically symmetric metric in advanced coordinates is given as \cite{Senov}:
\beq ds^2 = - e^{2\beta(v,r)} \left(1 - \frac{2m(v,r)}{R(r)}\right) dv^2 + 2 e^{\beta(v,r)}dv dr + R(r)^2 d \Omega^2, \label{SphSymMtrc} \eeq

\noindent where $m(v,r)$ is the mass function, $\beta(v,r)$ is an arbitrary function, and $d\Omega^2$ is the line-element for the 2-sphere. The remaining function $R(r)$ is of the form
\beq R(r) = \sqrt{r^2 + c^2}, \label{RegFun} \eeq 
\noindent where the parameter $c$ is a non-negative real-valued constant that determines if the metric is regular. In this form, all gauge freedom has been used, and in general, further simplification of the Einstein tensor is not possible. For example, in a perfect fluid solution, the fluid (or dust) will not, in general, be comoving.

We choose the two future pointing, radial, null, geodesic, contravariant vector fields as part of the frame basis 
\beq \ell = \partial_v + \frac12 \left( 1 - \frac{2m}{R}\right) e^{\beta}  \partial_r ,~~n = -e^{-\beta} \partial_r,  \label{SSln} \eeq
\noindent and complete the non-coordinate basis using the complex contravariant spatial vector

\beq m = \frac{1}{\sqrt{2}R} \partial_\theta + \frac{i}{\sqrt{2} R \sin(\theta) } \partial_\phi.  \label{SSmmb} \eeq
and its complex conjugate $ \bar{m} $. Relative to this null frame, the future null expansions are:
\beq \theta_{(\ell)} = \frac{e^{\beta} R_{,r} }{R}\left(1-\frac{2m}{R}\right),~~~ \theta_{(n)} = - \frac{2e^{-\beta} R_{,r}}{R}. \eeq

If the line-element represents a spherically symmetric black hole solution, the apparent horizon is defined according to equation \eqref{FOTH}.
If instead, we suppose the line-element describes a wormhole, then the wormhole throat is characterized by the two conditions in equation \eqref{HVthroat}. In both cases, the vanishing of $\theta_{(\ell)}$ implies that the relevant hypersurface for either the apparent horizon or wormhole throat is described by the equation $R(r)-2m(v,r) = 0$. We will denote this hypersurface as $\mathcal{H}$.


The gradient of the equation $R-2m=0$ yields the dual of the vector normal to this hypersurface:

\beq N_a = \nabla_a(R-2m) = (R_{,r}-2m_{,r}) dr - 2m_{,v} dv.  \eeq

The hypersurface will be timelike or spacelike, depending on the sign of the magnitude of the normal vector, or null if the magnitude is zero. Evaluating the norm on the surface and using equation \eqref{RegFun}, we obtain:
\beq |N|\big |_{\mathcal{H}}= g^{ab}N_a N_b\big |_{\mathcal{H}}= -4e^{-\beta}m_{,v}\left( \frac{r}{\sqrt{r^2+c^2}}-2m_{,r} \right) \bigg|_{\mathcal{H}}. \label{Normalnorm} \eeq

\noindent
To ensure that the resulting solution is dynamical, we assume that $m_{,v} \neq 0$. Then, the surface $\mathcal{H}$ will be spacelike, timelike, or null depending on the sign of $m_{,v}$ and the value of $m_{,r}$ on this surface.

The existence of the wormhole throat or apparent horizon affects the structure of the curvature tensor and its covariant derivatives. This is reflected in the vanishing of a particular scalar curvature invariant constructed from higher order SPIs  \cite{Karlhede:1982fj, Lake:2003qe}. An invariant originally constructed to detect the dynamical horizon for dynamical spherically symmetric black holes will also detect the wormhole throat \cite{Coley:2017vxb}.

\begin{thm}
For any spherically symmetric wormhole metric, the  wormhole throat $R(r)=2m(v,r)$  is detected by the vanishing of the first order SPI: 
\beq J \equiv 4 I_1 I_3  - I_5. \label{Jinvar}  \eeq 

\noindent where $I_1 = C_{abcd} C^{abcd},~~I_3 = C_{abcd;e} C^{abcd;e}$, and $I_5 = I_{1,a}I_{1}\!^{\,,a}$. 
\end{thm}

\begin{proof}
Relative to the coframe $\{ n,\ell, \bar{m}, m\}$ given in equations \eqref{SSln} and \eqref{SSmmb} we can compute the SPIs explicitly and show that:
\beq \theta_{(\ell)} \theta_{(n)} \propto  \frac{4 I_1 I_3  - I_5}{I_1^2} = \frac{J}{I_1^2}. \nonumber \eeq 
\noindent This relationship is invariant under boosts and spins of the complex null frame. 
\end{proof}

\noindent We may normalize the above SPI to produce a new invariant, $\tilde{J}$, whose vanishing is necessary and sufficient to detect a wormhole throat or apparent horizon: \beq \tilde{J} = \frac{J}{2^4 3 I_1^2} = \theta_{(\ell)} \theta_{(n)}. \label{Jtilde} \eeq

In \cite{Coley:2017vxb}, the authors have argued that the vanishing of the SPI $J$ determines the GH for spherically symmetric black hole metrics. However, this invariant will vanish on the throat of a wormhole, which is certainly not a black hole solution. This distinction requires a refinement of the geometric horizon detection conjectures to discriminate between wormholes and black holes. In principle, a wormhole throat and GH could be distinguished by the flare-out condition or the trapping condition for wormhole throats and apparent horizons, respectively.

Computing the norm squared of the gradient of $\tilde{J}$, we have a new curvature invariant constructed from the ratio of first order and second order SPIs: 
\beq |\nabla \tilde{J}|^2 \big |_{\mathcal{H}} =  -\frac{r^4 m_{,v} e^{-\beta}}{(r^2 + c^2)^{5}} \left ( \frac{r}{\sqrt{r^2 +c^2}} - 2 m_{,r} \right)  \bigg|_{\mathcal{H}}. \eeq


\noindent We will show that this curvature invariant will partially determine the flare-out condition on the throat. In order to fully determine the flare-out or trapping condition, we must use an additional scalar curvature invariant.

\section{Calculation of Cartan invariants} \label{sec:2}

To provide an invariant characterization of the conditions that define a HV wormhole throat or an apparent horizon, we will apply the CK algorithm \cite{GANG} to generate a set of curvature invariants that fully characterize the geometry. Relative to the coframe given by equations \eqref{SSln} and \eqref{SSmmb}, the only non-zero component of the Weyl spinor is:
\beq \Psi_2 = -C_{abcd} \ell^a m^b n^c \bar{m}^d \label{SSpsi2} \eeq	



\noindent This implies that the parameters of the null rotations about $\ell^a$ and $n^a$ are fixed to identity and the remaining frame freedom consists of boosts and spins. Generally, the boost parameter can be fixed at zeroth order as well, since the non-zero Newman-Penrose (NP) curvature scalars for the Ricci spinor are: 
\beq \Phi_{00} = \frac12 R_{ab}\ell^a \ell^b,~~ \Phi_{11} = \frac14 R_{ab}(\ell^a n^b + m^a \bar{m}^b),~~\Phi_{22} = \frac12 R_{ab} n^a n^b,~~ \Lambda = \frac{R^a_{~a}}{24} \eeq

\noindent Thus, the Weyl tensor is of algebraic type {\bf D} and the Ricci tensor is generally of algebraic type {\bf I} ($\Phi_{00} \neq 0$) relative to the alignment classification \cite{classb}.

At zeroth order, the linear isotropy group of the Riemann tensor consists of spins \footnote{In fact, the spins belong to the isotropy group of the metric, so all higher covariant derivatives of the Riemann tensor are invariant under spins.} and potentially a boost if the Ricci tensor is of Segre type $[(11)(1,1)]$ or $[(111,1)]$ \cite{kramer}. If the linear isotropy group is two-dimensional, then the boost parameter can be fixed using the components of the covariant derivative of the Weyl tensor. 

Applying a boost with parameter $a=a(v,r)$ to the null coframe in equations (\ref{SSln}) and (\ref{SSmmb}) gives: 
\beq \ell' = a^2 \ell,~~ n' = a^{-2} n,~~ m' = m. \eeq

\noindent We can use the differential Bianchi identities to simplify the components of the covariant derivative of the Weyl tensor to express them in terms of the zeroth order NP scalars, the frame derivatives of the zeroth order NP scalars with respect to $\ell'$ and $n'$, and the following spin coefficients:

\beq  \epsilon' &=& a^2 \left[ \frac{R(R-2m) e^\beta \beta_{,r}}{2R^2} + \frac{\beta_{,v} }{2} - \frac{e^\beta m_{,r}}{2R} +\frac{R_{,r} e^{\beta} m }{2 R^2}  \right] \label{SSeps} \\ 
& & + a^2 D( \ln a), \nonumber \\
 \gamma' &=& \frac{1}{a^2} \Delta (\ln a),  \label{SSgam} \\
 \rho' &=& -\frac12 \theta_{(\ell')} = -a^2 \frac{ R_{,r} e^\beta(R-2m)}{2R^2},  \label{SSrho} \\ 
 \mu' &=& \frac12 \theta_{(n')}  = - \frac{R_{,r} e^{-\beta}}{a^2 R},  \label{SSmu} \eeq


\noindent where $D = \ell^a \partial_{x^a}$ and $\Delta =  n^a \partial_{x^a}$ are frame derivatives of the original unprimed frame. The boost parameter is specified by normalizing the components of the Ricci tensor using a boost, 
\beq & \Phi_{00}' = a^4 \Phi_{00},~ \Phi_{22}'  = a^{-4} \Phi_{22}, & \eeq
\noindent or if the Ricci tensor is of Segre type $[(11)(1,1)]$ or $[(111,1)]$, a boost can be used to normalize the components of the covariant derivative of the Weyl tensor,
\beq C_{1231;4}' = a^2 C_{1231;4},~C_{1232;4}' = a^{-2} C_{1232;4}, \eeq

Once the Ricci tensor or Weyl tensor is chosen to fix the boost parameter, there are two subcases that can occur. If both components are non-zero, the boost may be chosen to make them equal in the new frame. If one component vanishes, the boost parameter may be used to set the other component equal to 1. Hereafter, we will assume an appropriate normalization has been made and omit the primes on all NP quantities.

By fixing the boost parameter, an invariantly defined frame has been constructed. Relative to this frame, the components of the Ricci and Weyl tensors are now Cartan invariants. Furthermore, the components of the respective covariant derivatives of the Ricci and Weyl tensors are also Cartan invariants and we may express the following spin-coefficients as ratios of Cartan invariants:

\beq \rho = \frac{C_{1232;4}}{3\Psi_2}, \mu = \frac{C_{1231;4}}{3\Psi_2}. \eeq

\noindent The two NP spin coefficients are now Cartan invariants relative to the invariantly defined frame based on our choice of how the boost paramter is fixed. Finally, any frame derivative of Cartan invariants involving $D$ or $\Delta$ will be a Cartan invariant.

To show that the remaining conditions for an apparent horizon and HV wormhole throat can be given in terms of Cartan invariants, we consider a boost $\ell = \tilde{a}^2 \tilde{\ell}$, where ${\tilde{\ell}}$ is the geodesic null direction for which $ -2 \tilde{D} \tilde{\rho} = \tilde{D} \theta_{(\tilde{\ell})} \geq 0$ and $\tilde{\rho}$ always vanishes on the geometric horizon. The Leibniz rule implies that $D \rho \geq 0$ if and only if $ \tilde{D} \tilde{\rho}\geq 0$. It follows that $-D \rho$ is proportional to the standard flare-out condition on the geometric horizon and shares the same sign.

From equation \eqref{SSrho}, the conditions for an apparent horizon in \eqref{apparent def} can be restated in terms of the Cartan invariants:
\beq \rho = 0,~ \mu >0 \text{, and } \Delta \rho > 0, \label{CK_AH_conditions}\eeq

\noindent while the conditions for a HV wormhole throat in \eqref{HVthroat} become
\beq \rho = 0 \text{ and } \ D \rho \leq 0. \label{CK_HV_conditions} \eeq

\noindent Both the HV wormhole throat and the apparent horizon are geometric surfaces as they are characterized by conditions on curvature invariants, namely the Cartan invariants.

As a final remark, the flare-out condition for the HV wormhole is a property of the bundle of out-going null geodesics and not a property of the geometry of the solution itself \cite{Hochberg:1998ha}. In the context of the CK algorithm, the invariant null vector field $\ell^a$ is no longer required to be geodesic and the analogous flare-out condition now arises from the geometry of the solution. We note that in the spherically symmetric case, the invariant null vector fields are geodesic but are no longer affinely parametrized. Despite this difference in interpretation, the HV wormhole conditions can be recovered from the formulation in terms of Cartan invariants.

\section{ Ratios of scalar polynomial curvature invariants} \label{sec:3}

The characterization of the apparent horizon and HV wormhole throat in terms of the Cartan invariants given, respectively, in equations \eqref{CK_AH_conditions} and \eqref{CK_HV_conditions}  has the advantage of generating invariants that will only detect the wormhole throat. However, it has the disadvantage that it relies on the choice of a particular class of coframes. It is preferable to determine invariant conditions for the wormhole throat in terms of SPIs or ratios of SPIs, which we will call scalar rational curvature invariants (SRIs), as these are independent of the choice of coframe. To accomplish this, we will express the SRIs in terms of Cartan invariants. For example, the SPI $J$ from \eqref{Jinvar} takes the form: 

\beq J =  (2^{12}) (3^3) \rho \mu \Psi_2^4 = (2^{11}) (3^3) \Psi_2^4 \left( \frac{r^2}{2 (r^2 + c^2)} \right) \left( 1 - \frac{2 m}{\sqrt{r^2 + c^2}} \right) . \label{JCartan} \eeq


To begin, we will determine a basis of SRIs at zeroth order by taking ratios of the following non-zero Carminati-McLenaghan (CM) invariants \cite{carminati1991algebraic}:

\beq m_1  &=& 2 \Psi_2 \Phi_{22} \Phi_{00} - 8 \Phi_{11}^2 \Psi_2, \\
m_2 &=& 2 \Psi_2^2 \Phi_{22} \Phi_{00} + 16 \Psi_2^2 \Phi_{11}^2,\\
r_2 &=& 6 \Phi_{11} \Phi_{00} \Phi_{22},  \\
w_1 &=& 6 \Psi_2^2, \\
w_2 &=& -6 \Psi_2^3.\eeq

\noindent The SRIs we will use are

\beq W_0&=&\frac{w_2}{w_1} = -\Psi_2, \\
R_1&=&\frac{2m_1}{W_0} + \frac{2m_2}{W_0^2} = 48 \Phi_{11}^2, \\
R_2&=&\frac{R_1}{6} - \frac{m_1}{W_0} = 2 \Phi_{00}\Phi_{22}, \\
R_3&=&\frac{r_2}{3 R_2} = \Phi_{11}.  \eeq

\noindent We note that the Cartan invariants $\Psi_2$, $\Phi_{00}\Phi_{22}$, and $\Phi_{11}$ are proportional to $W_0, R_2$, and $R_3$ respectively. We have included $R_1$ to cover the case where either $\Phi_{00}$ or $\Phi_{22}$ vanishes. In this case, we are only able to construct SRIs that are proportional to the Cartan invariants $\Psi_2$ and $\Phi_{11}^2$. 

In order to construct a helpful basis of SRIs that are in terms of components of the covariant derivatives of the Weyl and Ricci tensors, we choose a different invariant frame by fixing the boost parameter to normalize: 
\beq \mu = -1. \label{Alt frame fixing} \eeq

\noindent  This is equivalent to fixing the component, $C_{1231;4} = -3 \Psi_2$. In doing so, the spin coefficient $\rho$ is fixed to 
\beq \rho = - \tilde{J}, \eeq


\noindent where $\tilde{J}$ is the SRI from \eqref{Jtilde}. From  the transformation rule for $\rho$, given in \eqref{SSrho}, and relative to this frame, the corresponding Cartan invariant will still vanish on the wormhole throat. With this choice of frame, $-D \rho$ is still proportional to the standard flare-out condition on the geometric horizon and shares the same sign.

Using the Ricci identities, we find that two  spin-coefficients can be expressed in terms of SRIs and the Cartan invariant $\Phi_{22}$:
\beq \epsilon &=& \frac12\left(W_0  -\tilde{J}   \right) - 2 \Lambda, \label{RIep} \\
\gamma &=& \frac12 + \frac12 \Phi_{22}. \label{RIgam} \eeq

\noindent With these expressions, we can write $D\rho$ and $\Delta \rho$ in terms of SRIs and the Cartan invariants $\Phi_{00}$ and $\Phi_{22}$: 

\beq D \rho &=& -\tilde{J} W_0 + 2 \tilde{J}^2 - 2 \tilde{J} \Lambda + \Phi_{00}, \label{RIDrho} \\
\Delta \rho &=& - \tilde{J} \Phi_{22} + W_0 - 2 \tilde{J} - 2 \Lambda. \label{RIDelrho} \eeq

\noindent The norm squared of the exterior derivative of $\tilde{J}$ is 
\beq |\nabla \tilde{J}|^2 \propto D \tilde{J} \Delta \tilde{J}, \label{NormGradJt} \eeq

\noindent and it will give an SRI that encodes information about $D\rho$ and $ \Delta \rho$. We are only able to distinguish the relative sign between these Cartan invariants from the sign of $|\nabla \tilde{J}|^2$. Yet again, we will have to find another way to write the Cartan invariants in terms of SRIs. 

If either of the components $\Phi_{00}$ or $\Phi_{22}$ vanishes, then either equation \eqref{RIDrho} or \eqref{RIDelrho} yields an expression for $D\rho$ or $\Delta \rho$, respectively, in terms of SRIs. In the remainder of this section we will assume that both $\Phi_{00}$ and $\Phi_{22}$ are non-zero relative to the chosen frame. 

We now construct a vector field, denoted ${\bf P}_0$, by taking a linear combination of gradients of the zeroth order SRIs $W_0$,  $\Lambda$, and $R_3$. The resulting field is simplified via the Bianchi identities to give:

\beq {\bf P}_0 &=& \cfrac{1}{-2^{11} 3^2 W_0^2 }\nabla (48 W_0^2) + \cfrac{1}{2^6 3^1 W_0} \nabla \Lambda - \nabla R_3 \\
&=& -\cfrac{2R_3 + \tilde{J} \Phi_{22} + 3 W_0}{2^6 3^1  W_0} \: n + \cfrac{2 \tilde{J} R_3 - \Phi_{00} + 3 \tilde{J} W_0 }{2^6 3^1 W_0} \: \ell. \eeq

\noindent Taking the norm of ${\bf P}_0$ yields a first order SRI. Using equations \eqref{RIep} and \eqref{RIgam} and subtracting multiples of $\tilde{J}$ with zeroth order SRIs, we find a simpler first order SRI: 

\beq  (2 R_2 + 3 W_0) (\tilde{J}^2 \Phi_{22} + \Phi_{00}) &=& |P_0|^2 + \frac12 \tilde{J} R_2 + 9 \tilde{J} W_0^2 + 12 \tilde{J} W_0 R_3 \nonumber \\ && + 4 \tilde{J} R_3^2. \eeq

\noindent Assuming $2 R_2 + 3 W_0 \neq 0 $, we can express a linear combination of the Cartan invariants $\Phi_{22}$ and $\Phi_{00}$ in terms of SRIs:

\beq  J_1 = \tilde{J}^2 \Phi_{22} + \Phi_{00} =  \frac{|P_0|^2 + \frac12 \tilde{J} R_2 + 9 \tilde{J} W_0^2 + 12 \tilde{J} W_0 R_3 + 4 \tilde{J} R_3^2}{2 R_2 + 3 W_0}.  \eeq

If $2R_2 + 3 W_0 = 0$,  we can instead consider another expression for $J_1$, assuming $2  \tilde{J} - W_0 + 2 \Lambda \neq 0$:

\beq   J_1 = \frac{|\nabla \tilde{J}|^2 - 8  \tilde{J}^3 - \tilde{J}^2( 8 W_0 + 16 \Lambda  ) - \tilde{J}( 2 W_0^2 + R_2 + 8 W_0 \Lambda + 8 \Lambda^2 )}{2  \tilde{J} - W_0 + 2 \Lambda}.  \eeq

\noindent To exclude the possibility that both quantities in the denominator of $J_1$ vanish, we write them in their coordinate expressions to show that,
\beq 2 \tilde{J} - W_0 + 2 \Lambda = - (2 R_2 + 3 W_0) + \frac{1}{r^3} \nonumber \eeq
\noindent and note that both quantities cannot be zero simultaneously. Thus it is always possible to determine $J_1$.

With $J_1$, we can consider the combination of Cartan invariants $D\rho - \tilde{J} \Delta \rho$ and show that it is equal to a linear combination of SPIs: 

\beq J_2 = D\rho - \tilde{J} \Delta \rho &=& \tilde{J}^2 \Phi_{22} + \Phi_{00}  - 2 \tilde{J} W_0 + 4 \tilde{J} \Lambda + 4 \tilde{J}^2 \nonumber \\
&=& J_1 - 2 \tilde{J} W_0 + 4 \tilde{J} \Lambda + 4 \tilde{J}^2.  \eeq 

\noindent Noting that on the geometric surface we have $\tilde{J} |_{\mathcal{H}} = 0$, it follows that:

\beq J_2 |_{\mathcal{H}}  =  J_1 |_{\mathcal{H}} = D \rho |_{\mathcal{H}}   = - D \tilde{J}|_{\mathcal{H}}. \label{J2onsurface} \eeq

In the construction of $J_1$ and $J_2$, we have assumed that $\Phi_{00}$ and $\Phi_{22}$ are both non-zero. If one of these components vanishes, then we may use either equation \eqref{RIDrho} or \eqref{RIDrho} to generate an expression for either $D \tilde{J}$ or $\Delta \tilde{J}$:

\beq \begin{aligned}  D \tilde{J} &= \tilde{J} W_0 - 2 \tilde{J}^2 2 \tilde{J} \Lambda - \Phi_{00},  \\
\Delta \tilde{J} &= \tilde{J} \Phi_{22} - W_0 + 2 \tilde{J} + \Lambda. \end{aligned} \eeq

\noindent Then, using $|\nabla \tilde{J}|^2$, we are able to construct an SRI for the remaining quantity.

\section{The geometric surface conjectures for wormholes and black holes} \label{sec:4}

In the previous section we have shown that for any spherically symmetric metric describing a black hole or a wormhole, the geometric surface $\mathcal{H}$, defined by the vanishing of $\tilde{J}$ in \eqref{Jtilde}, can be characterized as an apparent horizon or as a HV wormhole throat by determining the signs of both $D \tilde{J}$ and $\Delta \tilde{J}$ using $|\nabla \tilde{J}|^2$ in equation \eqref{NormGradJt} and either $J_2$ in equation \eqref{J2onsurface} or the appropriate SRI expression from equations \eqref{RIDrho} and \eqref{RIDrho} when either $\Phi_{00}$ or $\Phi_{22}$ vanishes. If the metric describes a dynamical black hole, we have the following inequality: 
\beq \tilde{J}|_{\mathcal{H}} = 0 \text{ and }   \Delta \tilde{J} |_{\mathcal{H}} < 0. \eeq
\noindent If the metric describes a wormhole with a HV throat the inequality is instead
\beq \tilde{J}|_{\mathcal{H}} = 0 \text{ and } D \tilde{J} |_{\mathcal{H}} \geq 0. \eeq


This distinction between the apparent horizon and the HV wormhole throat suggests the following refinement for the geometric horizon conjecture: 

\begin{con}
A geometric horizon (GH) is a geometric surface defined by the vanishing of a curvature invariant proportional to $\theta_{(\ell)}$, i.e., $\tilde{I}_0 \propto \theta_{(\ell)}$, $\tilde{I}_0 =0$ relative to the invariantly defined complex null frame $\{ \ell, n, m, \bar{m} \}$, along with the inequalities on two curvature invariants $\tilde{I}_1  = \theta_{(n)}$ and $\tilde{I}_2 = \Delta \theta_{(\ell)} $:
\beq \tilde{I}_1 < 0 \text{ and } \tilde{I}_2 < 0. \eeq
\end{con}

In the same way that we expect dynamical black hole solutions to eventually settle down to either a black hole with an isolated horizon (IH) or a stationary black hole \cite{AshtekarKrishnan}, it is feasible that a black hole solution could transition into a wormhole solution or vice versa \cite{hayward1999dynamic}. In the case of a black hole transitioning to a stationary state, we would expect the curvature invariants which characterize the GH to smoothly track the GH as it evolves into an IH or a Killing horizon. In a similar manner, the curvature invariants that detect the GH of a black hole which then transitions into a wormhole (or vice versa) should also be able to distinguish the wormhole throat.

In analogy with the definition of a geometric horizon,  we will introduce an invariantly defined surface for wormhole throats. We will replace the condition that the outgoing and ingoing null directions $\ell^a$ and $n^a$ are geodesic and normal to the wormhole throat with the condition that $\ell^a$ and $n^a$ are invariantly defined. 

\begin{con}
An invariant Hochberg-Visser (IHV) wormhole throat is a geometric surface defined by the vanishing of a curvature invariant proportional to $\theta_{(\ell)}$, i.e., $\tilde{I}_0 \propto \theta_{(\ell)}$, $\tilde{I}_0 =0$ relative to the invariantly defined complex null frame $\{ \ell, n, m, \bar{m} \}$, along with the inequalities on two curvature invariants $\tilde{I}_1 = \theta_{(n)}$ and  $\tilde{I}_3 = D \theta_{(\ell)} $:
\beq \tilde{I}_1 < 0,~ \text{ and } \tilde{I}_3 \geq 0. \eeq
\end{con}

In addition to the GH and the IHV wormhole throats, there are other surfaces which are defined as the zero-set of curvature invariants and are physically or geometrically important. For example, the ergosphere in Kerr spacetimes can be determined by the zero-set of certain SRIs constructed from the Weyl tensor \cite{GANG}. In order to take into account for these surfaces, we introduce a general conjecture for physically relevant surfaces: 

\begin{con}
If a surface has significance in the physical interpretation of a solution to a gravity theory then it is characterized by a GS with additional conditions imposed on the curvature invariants.   
\end{con}

\subsection{Regular Vaidya black holes, traversable wormholes, and black bounces}

In this section, we will consider a novel class of dynamical spherically symmetric metrics describing transitions between regular Vaidya spacetimes, traversable wormholes, and black bounces \cite{Lobo:2020ffi, Simpson:2019cer}. We will only examine examples with ingoing radiation. In this case, the metric is constructed by fixing $\beta(v,r)=0$ and $m(v,r)=m(v)$ in \eqref{SphSymMtrc}. Explicitly, this gives

\beq ds^2 = -  \left(1 - \frac{2m(v)}{\sqrt{r^2+c^2}}\right) dv^2 + 2 dv dr + (r^2+c^2) d \Omega^2, \label{SMMVmetric} \eeq

\noindent where the mass function $m$ is dependent only on the null coordinate $v$, and $c$ is a positive real-valued constant. When $m(v)<\frac{c}{2}$ this metric describes a regular black hole and $c$ determines the size of the non-singular core. When $m(v)>\frac{c}{2}$, this metric describes a wormhole and $c$ determines the throat radius. In the case that $m(v)$ is increasing in $v$ and crosses the value $\frac{c}{2}$, the metric will transition from a regular black hole to a wormhole due to an accretion of null dust. We note that in the outgoing case the metric can describe an evaporating black hole which leaves a wormhole remnant. 

We will determine the geometric surfaces associated with the stages that describe either a wormhole throat or an apparent horizon at particular values of $r$ and $v$. To compute the Cartan invariants for these solutions we would start with the coframe given by equations \eqref{SSln} and \eqref{SSmmb}, with the appropriate simplifications. Since the matter content has $\Phi_{00}$ and $\Phi_{22}$ as non-zero, the boost parameter is chosen to normalize $\Phi_{00} = \Phi_{22}$. Instead of the Cartan invariants, we will compute the SRIs from the previous section:


%

\beq \begin{aligned} \tilde{J} & = \frac{r^2 (\sqrt{r^2+c^2} - 2 m)}{2 \sqrt{r^2+c^2} (r^2+c^2)^2}, \\ 
J_1 & = \frac{r^2(2m_{,v} (r^2 + c^2) r \sqrt{r^2 +c^2} -4 c^2 m(m - \sqrt{r^2+c^2}) - c^4 -c^2 r^2 )}{(r^2+c^2)^5},  \\ 
J_2 & = J_1 - 2 \tilde{J}( W_0 + 4 \Lambda + 4 \tilde{J}), \end{aligned} \eeq

\noindent where 

\beq W_0 = -\frac{3 m(c^2 - 2 r^2 )\sqrt{r^2+c^2} -2c^4 -2 c^2 r^2}{6(r^2 + c^2)^3},~~\Lambda = \frac{(3m- \sqrt{r^2+c^2})c^2 }{12 \sqrt{r^2+c^2} (r^2+c^2)^2} . \eeq

\noindent The familiar geometric surface, $\mathcal{H}$, is in this case defined by the solution set to $\tilde{J} = 0$. Explicitly, this occurs when

\beq r_{\mathcal{H}} (v) = \pm \sqrt{4m^2 - c^2}. \eeq

\noindent The signs of $J_2|_{\mathcal{H}}$ and $|\nabla \tilde{J}|^2|_{\mathcal{H}}$ will determine if $\mathcal{H}$ is a wormhole throat or geometric horizon.


We will also point out the existence of one other important hypersurface, $\mathcal{O}$, located at $r=0$. This hypersurface has an induced 3-metric of the form

\beq ds^2|_{\mathcal{O}} = -  \left(1 - \frac{2m(v)}{c}\right) dv^2 + c^2 d \Omega^2, \label{SMMV3metric} \eeq

\noindent which describes a cylinder with three possible signatures depending on the ratio:
\begin{itemize}
\item $\frac{2m}{c}<1$: $\mathcal{O}$ is timelike and is a traversable wormhole.
\item $\frac{2m}{c}=1$: $\mathcal{O}$ is null and is a one-way wormhole with a null throat.
\item $\frac{2m}{c}>1$: $\mathcal{O}$ is spacelike and is a black-bounce.  
\end{itemize} 

\noindent Here, a bounce describes the transition into a future incarnation of the Universe. We note that the nature of the hypersurface will change as $v$ varies.


The invariant $\tilde{J}$ will detect the hypersurface $\mathcal{O}$. However, $J_1$ and $J_2$ will vanish on $\mathcal{O}$ as well, and we cannot determine the signature of this surface. Due to the importance of this hypersurface in terms of its physical interpretation, we expect that its signature should be detected using some curvature invariant. This expectation is justified by considering the invariant $2\Lambda - W_0$ and evaluating it on the surface $r=0$:

\beq (2\Lambda - W_0) |_{\mathcal{O}} = \frac{2m-c}{2c^3}. \eeq 

\noindent The value of this invariant will determine whether $\mathcal{O}$ is a traversable wormhole, a one-way wormhole with a null throat, or a black bounce.

\newpage
\section{Conclusions} \label{sec:5}

We have shown that the throats of spherically symmetric dynamical wormholes defined by Hochberg and Visser (HV) are characterized by the same curvature invariant used to define the apparent horizon for spherically symmetric dynamical black holes. For spherically symmetric black holes, the apparent horizon is a geometric horizon (GH) and, in light of the GH conjectures \cite{Coley:2017vxb}, this is problematic, as a naive application of the conjectures might lead one to assume a HV wormhole throat is in fact a black hole's GH. This result indicates that additional curvature invariants are needed to distinguish between the relevant hypersurfaces for each of the two classes of solutions. 

By relaxing the condition that the ingoing and outgoing null directions must be geodesic and normal to the relevant surface in the definitions of the HV wormhole throat and the apparent horizon, we have given an alternative definition of these surfaces in terms of curvature invariants. The relevant curvature invariants are first constructed using an invariantly defined frame generated by the Cartan-Kalrhede algorithm and are hence Cartan invariants. However, as the calculation of the Cartan invariants relies on a particular choice of frame, we have also determined a set of rational curvature invariants constructed from the set of Carminati-McLenaghan invariants. This set of invariants will distinguish between the HV wormhole throat and the GH and has the added advantage that the curvature invariants are independent of the choice of frame basis. 

Inspired by the use of curvature invariants to distinguish between the HV wormhole throat and the GH in spherically symmetric dynamical metrics, we have introduced the concept of a geometric surface as a generalization of a GH and defined an invariant Hochberg-Visser (IHV) wormhole throat as another example of a geometric surface. For both the GH and IHV wormhole throats, we have required that the ingoing and outgoing directions $\ell^a$ and $n^a$ be invariantly defined, and we have relaxed the requirement that $\ell^a$ and $n^a$ must be geodesic and normal to the relevant surface. From these definitions, we have suggested three new conjectures on geometric surfaces (GS) as an extension of the current GH conjectures. 

We note that in the case of spherical symmetry, the HV wormhole throat and apparent horizon coincide with the IHV wormhole throat and GH, respectively. This may not be the case in less symmetric wormhole and black hole solutions, and it is expected that these surfaces will differ for more realistic solutions. For example, in the Kastor-Traschen multi-black hole solution, the GH and the apparent horizon are different surfaces \cite{NSH, AD}. In future work, we hope to investigate the validity of these conjectures for axisymmetric examples such as rotating dynamical black holes and  rotating traversable wormholes.

\section*{Acknowledgements}
We would like to Bahram Shakerin, Cooper Watson, Jeff Lee, and Eric Davis for their helpful discussions.

\bibliographystyle{unsrt-phys}
\bibliography{DynamicWormholes}

\end{document}